\documentclass[11pt,letterpaper]{article}
\usepackage[T1]{fontenc}
\usepackage[utf8]{inputenc}
\usepackage{lmodern}
\usepackage[margin=1in]{geometry}
\usepackage{amsmath,amssymb,amsthm,mathtools}
\usepackage{microtype}
\usepackage{xcolor}
\definecolor{linkblue}{RGB}{0,51,102}
\usepackage[pdfusetitle,
 bookmarks=true,bookmarksnumbered=false,bookmarksopen=false,
 breaklinks=true,pdfborder={0 0 0},pdfborderstyle={},backref=false,colorlinks=true]
 {hyperref}
\hypersetup{
 linkcolor=linkblue, urlcolor=linkblue, citecolor=linkblue}

\usepackage{tikz}
\usetikzlibrary{arrows.meta,decorations.pathreplacing}

\usepackage{graphicx}
\usepackage{caption} 
 
\usepackage{bookmark}
\allowdisplaybreaks[1]
\newtheorem{theorem}{Theorem}
\newtheorem{lemma}{Lemma}
\newtheorem{proposition}{Proposition}
\newtheorem{corollary}{Corollary}
\theoremstyle{definition}

\theoremstyle{remark}

\newcommand{\Tr}{\operatorname{Tr}}
\newcommand{\id}{\operatorname{id}}
\newcommand{\supp}{\operatorname{supp}}
\newcommand{\ran}{\operatorname{ran}}
\newcommand{\cL}{\mathcal L}
\newcommand{\cD}{\mathcal D}
\newcommand{\cE}{\mathcal E}
\newcommand{\ket}[1]{\lvert #1\rangle}
\newcommand{\bra}[1]{\langle #1\rvert}
\newcommand{\proj}[1]{\lvert #1\rangle\!\langle #1\rvert}
\newcommand{\ketbra}[2]{\lvert #1\rangle\!\langle #2\rvert}
\newcommand{\norm}[1]{\left\lVert #1\right\rVert}
\newcommand{\abs}[1]{\left\lvert #1\right\rvert}
\newcommand{\step}[1]{\par\medskip\noindent\textbf{#1}\quad\ignorespaces}
\newcommand{\authorname}{Mark M. Wilde}
\newcommand{\authoremail}{wilde@cornell.edu}
\title{\textbf{Strong converse for the quantum capacity\\of the pure-loss bosonic channel}}
\author{\texorpdfstring{\authorname\\[0.4em]
\small \textit{School of Electrical and Computer Engineering, Cornell University}\\
\small \textit{Ithaca, New York 14853, USA}\\
\small\texttt{\authoremail}}{\authorname}}
\date{}
\begin{document}
\maketitle
\begin{abstract}
This paper reports the proof of a strong converse for the unconstrained quantum capacity of the pure-loss bosonic channel. At every fixed rate above capacity, the entanglement-generation fidelity of every code is bounded by a constant times the reciprocal of the number of channel uses. The bound holds without an energy constraint and for arbitrary encoded states, including states correlated across all input modes, and arbitrary joint decoders.
The proof combines quantum Chebyshev and hockey-stick testing inequalities with a uniform relative-entropy-variance bound for the balanced pure-loss channel, corresponding to transmissivity $\eta=1/2$. The variance bound follows by expressing the balanced beam splitter in bright and dark modes: the dark modes are exactly in vacuum, and any state orthogonal to that vacuum contains at least one dark photon. For general transmissivity, dilating the degrading attenuator reduces the problem to this balanced-channel setting and bounds the decoder test by precisely the factor that produces the known quantum-capacity threshold. The resulting argument establishes the strong converse at the unconstrained quantum capacity for every pure-loss bosonic channel.
\end{abstract}

\tableofcontents

\section{Introduction}

The quantum capacity of a communication channel is the largest asymptotic
rate at which entanglement can be transmitted with fidelity tending to
one. A capacity formula by itself does not determine the performance of codes
whose rates exceed that value: their fidelities might remain bounded
away from both zero and one. A strong converse excludes this possibility.
It states that every sequence of codes operating above capacity has
fidelity tending to zero. This distinction is particularly important for
bosonic channels, whose infinite-dimensional input spaces permit code
states with arbitrarily high energy.

The pure-loss bosonic channel describes an input mode mixed with a vacuum
mode at a beam splitter of transmissivity $\eta \in \left[0,1\right]$,
followed by discarding one output. It is a basic
model of optical attenuation; see \cite{Shapiro2009} for physical
motivation underlying the channel model. Holevo and Werner developed the capacity
analysis of bosonic Gaussian channels~\cite{HolevoWerner2001}. Wolf,
P\'erez-Garc\'ia, and Giedke evaluated the quantum capacities of a class of
Gaussian channels that includes the pure-loss channel~\cite[Eq.~(12)]{Wolf2007}.
For $1/2<\eta<1$, its unconstrained quantum capacity is equal to
$\log_2[\eta/(1-\eta)]$ qubits per mode; for $\eta\leq1/2$, it is equal to zero.
Energy-constrained private and quantum capacities, and the coding and
analytic issues associated with infinite-dimensional systems, were
studied in~\cite{WildeQi2018}. In this paper, no energy
constraint is imposed on the code.

\paragraph{Literature review.} There is a substantial literature on quantum strong converses. Morgan
and Winter proved a ``pretty strong'' converse for degradable
channels: rates above capacity cannot have fidelity approaching one,
with a bound stronger than a weak converse but not a conclusion of
vanishing fidelity~\cite[Theorem~2]{MorganWinter2014}. Their analysis also
reduces a strong-converse question for degradable channels to one for
associated symmetric channels~\cite[Theorem~19]{MorganWinter2014}. 
Ref.~\cite{WildeWinterErasure2014} proved a strong converse for almost all codes of the
quantum erasure channel.
Ref.~\cite{TomamichelWildeWinter2017} obtained strong-converse rates for quantum
communication from a quantity known as the Rains information of a channel and identified channel classes for
which these rates coincide with capacity.
The distinction between a strong-converse upper rate and the actual
capacity is essential when interpreting such bounds.

Strong converses for other bosonic communication settings were established
in~\cite[Theorem~24 and Corollary~25]{WildeTomamichelBerta2017}.
That is, Ref.~\cite{WildeTomamichelBerta2017} proved the strong converse for the unconstrained two-way-assisted
quantum capacities of the pure-loss channel and the quantum-limited
(pure) amplifier. For pure loss, allowing local operations and unlimited
two-way classical communication (LOCC) gives the threshold
$-\log_2(1-\eta)$~\cite[Corollary~25]{WildeTomamichelBerta2017}.
This assisted threshold is strictly larger than the unassisted value
considered here for $0<\eta<1$. For the quantum-limited bosonic amplifier channel with
 gain $G>1$, the assisted threshold is
$\log_2[G/(G-1)]$~\cite[Corollary~25]{WildeTomamichelBerta2017}.
It coincides with the amplifier's unassisted quantum capacity
\cite[Eq.~(12)]{Wolf2007}, so their result also establishes the strong
converse for its unassisted quantum capacity.

Recent preprints establish substantially broader finite-dimensional results. Exponential strong converses for finite-dimensional degradable and antidegradable channels are presented in~\cite{Kondra2026}, while~\cite{ChengTomamichel2026} treats unassisted classical and quantum communication over arbitrary finite-dimensional memoryless channels. A complementary approach based on a fully quantum blowing-up lemma yields an exponential strong converse for quantum capacity~\cite{BeigiTomamichel2026}.
These recent finite-dimensional results provide useful context for the present work,
but none of their strong-converse theorems or proof techniques are used in
the argument below. The estimates reported here act on the
bosonic channel itself and remain independent of any photon cutoff.
Furthermore, the fidelity decay proved here is polynomial rather than exponential.

For optical classical communication, Ref.~\cite{WildeWinter2014} established a
strong converse for the pure-loss channel under a photon-number
occupation constraint, and this paper also explained the importance of distinguishing
that constraint from a mean-photon-number constraint.
Ref.~\cite{Bardhan2015} extended the classical
strong-converse analysis to optical Gaussian communication
channels. These results concern classical capacity and
specified energy constraints. The main theorem presented here concerns unassisted
quantum communication at the unconstrained quantum-capacity threshold.
It does not assert the sharper strong converse associated with a finite
energy-constrained quantum capacity.

\paragraph{Contribution and proof idea.}
We prove an explicit upper bound on the entanglement-generation fidelity
that applies to every encoded state and decoder. At a fixed
positive rate gap, the bound is $O(1/n)$, where $n$ is the number of modes.
The argument has two channel-specific ingredients.
First, receiver--environment symmetry at a balanced beam splitter makes
the relevant relative entropy equal to zero. We prove that the associated
relative-entropy variance is at most $4n$ for every pure reference--input
state with finite photon support. The bound is independent of that support,
the reference dimension, and entanglement across the input modes. Second, the degrading
isometry converts the state of a general pure-loss channel into 
that resulting from a balanced beamsplitter (see Figure~\ref{fig:balanced-decomposition}
in this context). Evaluating the decoder test on this balanced-channel
state gives a comparison trace controlled by the degrading attenuator,
without introducing the dimension of the extra bosonic output.

Part of the first ingredient relies on established prior work. Sharma and
Warsi explicitly state the logarithmic trace inequality and the quantum
Chebyshev inequality used below
\cite[Lemmas~5 and~6, Appendix~A]{SharmaWarsi2012}.
Their proof of the trace inequality builds on closely related
operator-monotone trace inequalities; see Audenaert et al.
\cite[Lemma~1]{Audenaert2007} and Petz
\cite[Theorem~11.18]{Petz2008}.
Cheng and Hsieh restated the quantum Chebyshev inequality, attributed it to
Sharma and Warsi, and combined it with a hockey-stick testing
bound~\cite[Lemma~20]{ChengHsieh2018}.
We reproduce the needed trace inequality and
its hypothesis testing consequence. We also retain an elementary refinement
obtained by centering the Neyman--Pearson projector. The strong converse would also follow
from the unrefined quantum Chebyshev inequality.

Relative-entropy variance is a familiar quantity in quantum
hypothesis testing, notably in the second-order analyses of Tomamichel
and Hayashi~\cite[Definition~1]{TomamichelHayashi2013} and
Li~\cite[Eq.~(1)]{Li2014}.
Here it is not evaluated on independent copies of a fixed output state.
The required estimate is uniform over correlated code states. The
balanced splitter supplies this control through an unused output-mode
combination that is prepared in the vacuum state. The unit photon-number gap above
that vacuum converts a logarithmic fluctuation into a commutator
estimate. Creating one photon in this vacuum mode and taking a Schmidt
decomposition bounds the commutator without any dependence on input
energy. This uniform Gaussian estimate is the main analytical step.

\paragraph{Paper organization.}
Section~\ref{sec:main} defines the code model and states the quantitative
converse. Section~\ref{sec:testing} gives the hockey-stick testing tool.
Section~\ref{sec:variance} proves the balanced-channel variance estimate.
Section~\ref{sec:transport} evaluates the decoder test on the balanced-channel
state and computes its comparison trace, and Section~\ref{sec:completion} completes the proof for arbitrary
normal inputs. Section~\ref{sec:conclusion} summarizes the result and
identifies further research questions. Appendix~\ref{app:trace} provides a detailed proof of the
logarithmic trace inequality, and Appendix~\ref{app:dark-parity} verifies
the identification of receiver--environment exchange with dark-mode
parity. Apart from the previously established capacity formula, the proof uses
finite-dimensional matrix analysis and elementary properties of Fock
space and trace-class operators. We include the operator inequalities and
limiting arguments needed below.

\section{Channel, code model, and main theorem}\label{sec:main}

\subsection{Notation and the pure-loss channel}
For much of this paper, Hilbert spaces are taken to be separable. States are positive trace-class
operators of trace one, and channels are completely positive,
trace-preserving maps on trace-class operators. Their adjoints are
unital completely positive maps on bounded operators.
For a multipartite state $\psi$, a subsystem subscript denotes
the corresponding reduced state; for example, $\psi_B\coloneqq\Tr_{RE}\psi_{RBE}$.
We write
$\norm{Z}_1\coloneqq\Tr\sqrt{Z^\dagger Z}$ and
$\norm{Z}_2^2\coloneqq\Tr Z^\dagger Z$ for the trace and Hilbert--Schmidt
norms, respectively. We use $\norm Z_\infty$ for the operator norm.
We use Dirac notation for vectors: $\ket{\psi}$ denotes a vector and
$\norm{\ket{\psi}}$ its Hilbert-space norm. When $\ket{\psi}$ is a unit vector,
the corresponding pure-state density operator is $\psi\coloneqq\proj{\psi}$.
The notation $Z\geq0$ denotes positive
semidefiniteness. All logarithms denoted by $\ln$ are natural logarithms;
communication rates use $\log_2$. The von Neumann entropy of a state $\omega$ is denoted by
\begin{equation}
H(\omega)\coloneqq-\Tr\omega\ln\omega. \label{eq:vNeu-ent}
\end{equation}
A channel notation such as $\cD\colon B^n\to\widehat R$ specifies its input
and output systems; the map acts on their trace-class operators.
Identity factors are omitted when their systems are unambiguous.

The one-mode Fock space has orthonormal basis $\{\ket m:m\geq0\}$ and
annihilation operator $a$, defined on number states by
$a\ket m\coloneqq\sqrt m\ket{m-1}$ for $m\geq1$ and $a\ket0\coloneqq0$.
Choose phases and define the beam-splitter isometry $V_t$ on this basis by
\begin{equation}
 V_t\ket m_A
 \coloneqq\sum_{k=0}^m\sqrt{\binom mk}\,
 t^{k/2}(1-t)^{(m-k)/2}\ket k_B\ket{m-k}_E,
 \qquad 0\leq t\leq1.
 \label{eq:isometry}
\end{equation}
Eq.~\eqref{eq:isometry} is obtained by expanding
$(\sqrt t\,b^\dagger+\sqrt{1-t}\,e^\dagger)^m/\sqrt{m!}$ on the output
vacuum. Indeed, the coefficient of $\ket k_B\ket{m-k}_E$ is
\begin{equation}
 \frac1{\sqrt{m!}}\binom mk t^{k/2}(1-t)^{(m-k)/2}
 \sqrt{k!(m-k)!}
 =\sqrt{\binom mk}t^{k/2}(1-t)^{(m-k)/2}.
 \label{eq:fock-coefficient}
\end{equation}
At $t=0$ and $t=1$, the coefficients are interpreted by continuity.
The binomial theorem implies that the norm of $V_t\ket m_A$ is equal to one, and distinct input number states
have orthogonal images because their total output photon numbers differ.
Thus $V_t$ extends to an isometry on the full Fock space.

The pure-loss channel $\cL_t$ is then defined by
\begin{equation}
\cL_t(Z)\coloneqq\Tr_E[V_tZV_t^\dagger],
\end{equation}
and tracing the other output
gives $\cL_{1-t}$ under the fixed identification of the mode spaces.

For later reference, define the capacity value
\begin{equation}
 q_\eta\coloneqq
 \begin{cases}
 0,&0\leq\eta\leq1/2,\\[1mm]
 \displaystyle\log_2\frac{\eta}{1-\eta},&1/2<\eta<1.
 \end{cases}
 \label{eq:qeta}
\end{equation}
Eq.~\eqref{eq:qeta} is the known unconstrained quantum capacity
\cite[Eq.~(12)]{Wolf2007}; see also the energy-constrained treatment
in~\cite{WildeQi2018}. Its achievability is not reproved here.
The converse below derives the same threshold directly from the
beam-splitter structure.

\subsection{Codes and fidelity}

Let quantum systems $R$ and $\widehat R$ have dimension $M\in \mathbb{N}$. Define
\begin{equation}
 \ket{\Phi_M}_{R\widehat R}
 \coloneqq\frac1{\sqrt M}\sum_{j=1}^M\ket j_R\ket j_{\widehat R},
 \qquad \Phi_M\coloneqq\proj{\Phi_M}.
 \label{eq:maximally-entangled}
\end{equation}
An $(n,M)$ entanglement-generation code consists of an encoded state
$\rho_{RA^n}$ and a decoder $\cD\colon B^n\to\widehat R$. Its fidelity $F$ is defined as
\begin{equation}
 F\coloneqq\Tr\!\left[
 \Phi_M(\id_R\otimes\cD\circ\cL_\eta^{\otimes n})(\rho_{RA^n})
 \right].
 \label{eq:code}
\end{equation}
We allow arbitrary $\rho_R$.
In particular, an arbitrary encoding channel
$\cE\colon\mathbb C^M\to A^n$ acting on one share of $\Phi_M$ produces
$\rho_{RA^n}\coloneqq(\id_R\otimes\cE)(\Phi_M)$ and is included in the model.
No assistance, preshared entanglement, or feedback is included. Arbitrary
joint encoding across the $n$ input modes and joint decoding of their
outputs are allowed.

A sequence of $(n,M_n)$ codes indexed by $n\in\mathbb{N}$ has rate at least $r$ if
\begin{equation}
\liminf_{n\to\infty}\frac{1}{n}\log_2M_n\geq r.
\end{equation}
A strong converse at $q_\eta$ asserts that a strict inequality between
this lower limiting rate and $q_\eta$ forces $F_n\to0$ for every such sequence of codes in the limit $n\to \infty$.

\begin{theorem}[Unconstrained fidelity bound]\label{thm:main}
Fix $0\leq\eta<1$ and $n,M\in\mathbb{N}$. Every
code as defined above satisfies, for every $b>0$,
\begin{equation}
F
\leq
\frac{4n}{4n+b^2(\ln 2)^2}
+
2^{b+nq_\eta-\log_2 M}.
\label{eq:main-free}
\end{equation}
In particular, define the rate gap
\begin{equation}
\delta_n\coloneqq \log_2 M-nq_\eta.
\label{eq:rate-gap}
\end{equation}
If $\delta_n>0$, then
\begin{equation}
F
\leq
\frac{16n}{16n+\delta_n^2(\ln 2)^2}
+
2^{-\delta_n/2},
\label{eq:main-gap}
\end{equation}
by setting $b = \delta_n/2$.
Consequently, for every $\Delta>0$, a code such that
$\log_2 M\geq n(q_\eta+\Delta)$ satisfies
\begin{equation}
F
\leq
\frac{16}{16+n\Delta^2(\ln 2)^2}
+
2^{-n\Delta/2}.
\label{eq:main-rate}
\end{equation}
These bounds are uniform over the encoded state, its energy, correlations
among the input modes, and the decoder.
\end{theorem}

\begin{corollary}[Strong converse]\label{cor:sc}
For every sequence of codes satisfying
$\liminf_{n\to\infty}\frac{1}{n}\log_2M_n>q_\eta$, the fidelities tend to zero as $n\to\infty$.
In particular, the unconstrained quantum capacity of the pure-loss
bosonic channel obeys the strong converse property.
\end{corollary}
\begin{proof}
Choose $\Delta>0$ strictly smaller than the gap between the lower limiting
rate and $q_\eta$. By the definition of the lower limit,
$\log_2M_n\geq n(q_\eta+\Delta)$ for all sufficiently large $n$.
Eq.~\eqref{eq:main-rate} then applies. Its first term tends to zero
as $1/n$, and its second term tends to zero exponentially, proving the
claim. The proof of Theorem~\ref{thm:main} is given below.
\end{proof}

\subsection{Why balanced loss suffices}\label{sec:balanced-overview}

At $\eta=1/2$, the receiver and the environment are interchangeable.
For an input with finite photon support, this symmetry makes the coherent
information, $H(B^n)-H(E^n)$, equal to zero. The converse requires control
of fluctuations as well. We will prove that the relevant relative-entropy
variance is at most $4n$, independently of the photon support. A quantum
Chebyshev inequality then bounds the acceptance probability of every test
on the output, including the test associated with a decoder.

For $1/2<\eta<1$, the receiver can be attenuated further to reproduce
the environment. Guha, Shapiro, and Erkmen identified the corresponding
degrading beam splitter: its transmissivity is given by
\cite[Fig.~1]{GuhaShapiroErkmen2008}
\begin{equation}
\tau\coloneqq\frac{1-\eta}{\eta}.
\end{equation}
Then the original pure-loss isometry followed by a dilation of
$\mathcal{L}_{\tau}$ has the following action on the input creation operator
$a^\dagger$:
\begin{equation}
a^\dagger
\longmapsto
\sqrt{1-\eta}e^\dagger
+\sqrt{1-\eta}e'^\dagger
+\sqrt{2\eta-1}f^\dagger.
\end{equation}
The equal coefficients of $e^\dagger$ and $e'^\dagger$ reveal an
equivalent factorization: first split the input into systems $G$ and
$F$ according to
\begin{equation}
a^\dagger
\longmapsto
\sqrt{2(1-\eta)}g^\dagger
+\sqrt{2\eta-1}f^\dagger,
\end{equation}
and then send $G$ through a balanced beam splitter,
\begin{equation}
g^\dagger
\longmapsto
\frac{e^\dagger+e'^\dagger}{\sqrt{2}}.
\end{equation}
Thus, for a pure encoded input on $RA^n$, the state immediately before
the balanced beam splitters is pure on $RF^nG^n$, and the final
three-output state is obtained by sending $G^n$ through
$V_{1/2}^{\otimes n}$ while treating $RF^n$ as the reference system.
This is the precise sense in which the resulting state is a
balanced-channel state with reference $RF^n$. See Figure~\ref{fig:balanced-decomposition} for a visualization.

\begin{figure}[t]
\centering
\resizebox{\linewidth}{!}{%
\begin{tikzpicture}[
  x=1cm,y=1cm,
  font=\small,
  line cap=round,line join=round,
  optical/.style={line width=.75pt,-{Latex[length=1.7mm,width=1.15mm]}},
  reference/.style={line width=.65pt,-{Latex[length=1.7mm,width=1.15mm]}},
  bs/.style={draw,line width=.55pt,fill=white,
             minimum width=5.6mm,minimum height=5.6mm,inner sep=0pt},
  heading/.style={font=\small\bfseries,anchor=west},
  port/.style={inner sep=2pt},
  parameter/.style={font=\small,anchor=south east,inner sep=1pt}
]

\begin{scope}
  \node[heading] at (-.20,2.03) {(a) Channel and degrader};
  \node[bs] (a1) at (1.35,0) {};
  \node[bs] (a2) at (4.00,0) {};

  \draw[optical] (0,0) node[port,left] {$A$} -- (a1.west);
  \draw[optical] (a1.east) -- node[above=3pt] {$B$} (a2.west);
  \draw[optical] (a2.east) -- (5.45,0) node[port,right] {$E'$};

  \draw[optical,rounded corners=3pt]
    (a1.north) -- (1.35,1.35) -- (5.45,1.35)
    node[port,right] {$E$};
  \draw[optical,rounded corners=3pt]
    (a2.south) -- (4.00,-1.35) -- (5.45,-1.35)
    node[port,right] {$F$};

  \draw[optical] (1.35,-.73) node[below=2pt] {$\lvert0\rangle$}
    -- (a1.south);
  \draw[optical] (4.00,.73) node[above=2pt] {$\lvert0\rangle$}
    -- (a2.north);

  \draw[line width=1.1pt] (a1.south west) -- (a1.north east);
  \draw[line width=1.1pt] (a2.north west) -- (a2.south east);
  \node[parameter] at (1.04,.36) {$\eta$};
  \node[parameter] at (3.69,.36) {$\tau$};

  \draw[reference] (0,-2.23) node[port,left] {$R$}
    -- (5.45,-2.23) node[port,right] {$R$};
\end{scope}

\node[font=\Large] at (6.55,-.14) {$\equiv$};

\begin{scope}[xshift=7.65cm]
  \node[heading] at (-.20,2.03) {(b) Balanced factorization};
  \node[bs] (b1) at (1.35,0) {};
  \node[bs] (b2) at (4.00,0) {};

  \draw[optical] (0,0) node[port,left] {$A$} -- (b1.west);
  \draw[optical] (b1.east) -- node[above=3pt] {$G$} (b2.west);
  \draw[optical] (b2.east) -- (5.45,0) node[port,right] {$E'$};

  \draw[optical,rounded corners=3pt]
    (b1.south) -- (1.35,-1.35) -- (5.45,-1.35)
    node[port,right] {$F$};
  \draw[optical,rounded corners=3pt]
    (b2.north) -- (4.00,1.35) -- (5.45,1.35)
    node[port,right] {$E$};

  \draw[optical] (1.35,.73) node[above=2pt] {$\lvert0\rangle$}
    -- (b1.north);
  \draw[optical] (4.00,-.73) node[below=2pt] {$\lvert0\rangle$}
    -- (b2.south);

  \draw[line width=1.1pt] (b1.north west) -- (b1.south east);
  \draw[line width=1.1pt] (b2.south west) -- (b2.north east);
  \node[parameter] at (1.04,.36) {$2(1-\eta)$};
  \node[parameter] at (3.69,.36) {$\tfrac12$};

  \draw[reference] (0,-2.23) node[port,left] {$R$}
    -- (5.45,-2.23) node[port,right] {$R$};

  \draw[decorate,decoration={brace,amplitude=4pt},line width=.55pt]
    (5.95,-1.18) -- (5.95,-2.40);
  \node[anchor=west,align=left,font=\footnotesize,inner sep=0pt]
    at (6.19,-1.79) {reference\\$RF$};
\end{scope}

\node at (7.00,-3.12) {$\displaystyle
  a^\dagger\longmapsto
  \sqrt{1-\eta}\,e^\dagger
  +\sqrt{1-\eta}\,e'^\dagger
  +\sqrt{2\eta-1}\,f^\dagger,
  \qquad \tau\coloneqq\frac{1-\eta}{\eta}.
$};
\end{tikzpicture}%
}
\caption{%
Equivalent vacuum-fed beam-splitter dilations for $1/2<\eta<1$.
The labels beside the beam splitters denote power transmissivities,
with $\tau\coloneqq(1-\eta)/\eta$. Mode phases follow the convention
of the displayed creation-operator identity.
(a) The pure-loss channel of transmissivity $\eta$ is followed by
the degrading attenuation of transmissivity $\tau$.
(b) An initial splitting of transmissivity $2(1-\eta)$ produces
$G$ and $F$, after which only $G$ enters the balanced beam splitter.
Both networks have the same labeled outputs $E,E',F$ and leave
$R$ unchanged. For $n$ uses, the optical splitters act modewise.
For a pure input on $RA^n$, the state before the balanced splitters
in panel (b) is pure on $RF^nG^n$, so $RF^n$ is the reference for
the balanced channel $\mathcal{L}_{1/2}^{\otimes n}$.}
\label{fig:balanced-decomposition}
\end{figure}
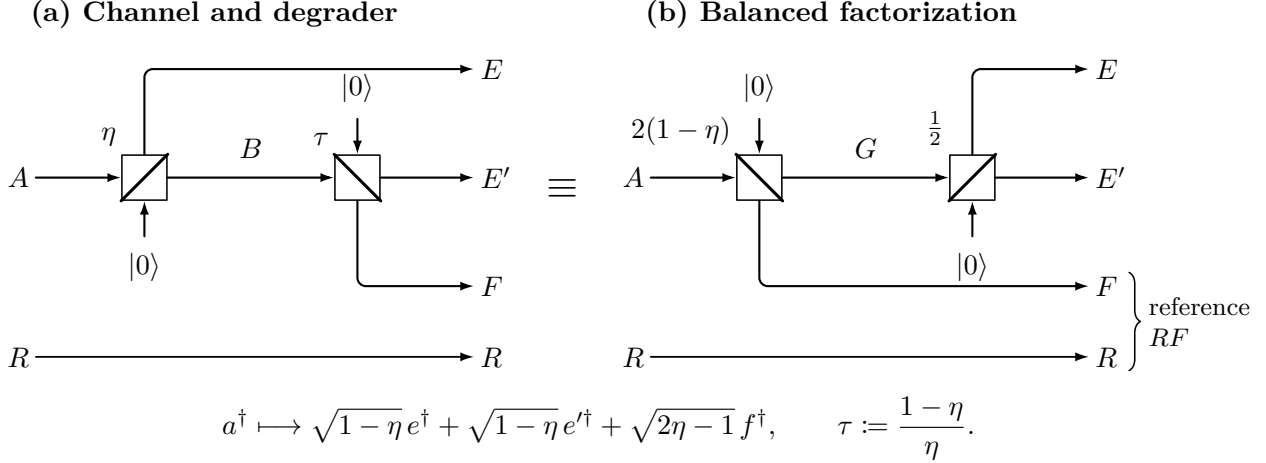

The decoder test can be evaluated on this three-output state without
changing the code's fidelity. The comparison weight in the testing
inequality is at most $[\eta/(1-\eta)]^n/M$, as proved in
Section~\ref{sec:transport}. This factor produces the capacity term in
\eqref{eq:main-free}. The enlarged system $RF^n$ serves as a reference
for that inequality; the logical target of the code remains $R$.
This distinction allows us to apply the balanced-channel estimate while
keeping the original message dimension $M$.

The factorization used here is a concrete
infinite-dimensional bosonic realization of Morgan and Winter's method
for reducing the strong converse of degradable channels to that for symmetric channel \cite[Section~VII]{MorganWinter2014}. After dilating the
degrading attenuator, an initial splitting $A\to GF$ is followed by the
balanced pure-loss isometry $V_{1/2}\colon G\to EE'$, which is unchanged
by exchanging $E$ and $E'$. Its reduced channel
$\mathcal L_{1/2}$ is therefore symmetric. Morgan and Winter's theorem
assumes finite-dimensional channels and is not invoked here; 
the bosonic factorization and the required cutoff-independent estimates
are established directly.

\subsection{Finite photon support for the operator calculations}
For $n$ modes, let $\Pi_L^{(n)}$ denote the projection onto the subspace
with total photon number at most $L\in\{0,1,2,\ldots\}$:
\begin{equation}
 \Pi_L^{(n)}\coloneqq\sum_{\substack{m_1,\ldots, m_n\geq 0 , \\ m_1+\cdots+m_n\leq L}}
 \proj{m_1,\ldots,m_n}.
 \label{eq:cutoff}
\end{equation}
It has finite rank, and $\Pi_L^{(n)}\to I^{\otimes n}$ strongly as $L\to\infty$.
We first prove all analytical estimates for pure input states supported on
$\Pi_L^{(n)}$, for an arbitrary finite $L$. Passive beam splitters preserve
total photon number, and so all output density operators then have finite
support. The annihilation operators remain the full physical Fock
operators; no truncated canonical commutation relation is used.

Every bound to be established will be independent of $L$. In
Section~\ref{sec:completion}, we remove this restriction for each fixed
code using convergence of fidelity alone. For this purpose, we do not need continuity of
the logarithmic moment on arbitrary infinite-energy states.

\section{A quantum Chebyshev bound in hockey-stick form}\label{sec:testing}

The coding argument requires an upper bound on the acceptance probability
of a test with a small comparison weight. We derive such a bound from
the relative entropy and its variance. Allowing an unnormalized comparison
operator is essential because the coding proof will use an operator of
the form $I_R\otimes\sigma_B$.

\subsection{Comparison operators and testing}
The developments in this section apply to finite-dimensional Hilbert spaces. For a Hermitian $H$, let
$H_+$ and $H_-$ be its positive and negative parts, so that
$H=H_+-H_-$ and $H_+H_-=0$. Write $\{H>0\}$ for its strictly positive
spectral projection. An effect is an operator $T$ that satisfies $0\leq T\leq I$.
For a state $\rho$ and a positive semidefinite comparison operator $S$, define
\begin{equation}
 E_\gamma(\rho\|S)\coloneqq\Tr(\rho-\gamma S)_+,
 \qquad\gamma>0.
 \label{eq:hockey}
\end{equation}
We use the unadjusted positive-part convention in~\eqref{eq:hockey},
including when $S$ is unnormalized. For normalized states and
$\gamma\geq1$, the above quantity reduces to the usual quantum hockey-stick divergence.

For every Hermitian $H$,
\begin{equation}
 \max_{0\leq T\leq I}\Tr TH=\Tr H_+.
 \label{eq:variational-positive-part}
\end{equation}
Indeed, $\Tr TH=\Tr TH_+-\Tr TH_-\leq\Tr H_+$, because all traces
of products of positive operators here are nonnegative and $T\leq I$.
Equality holds at $T\coloneqq\{H>0\}$. Therefore
\begin{equation}
 \Tr\rho T\leq E_\gamma(\rho\|S)+\gamma\Tr ST.
 \label{eq:hockey-test}
\end{equation}
This is the form in which the testing inequality enters the coding proof.
The quantity $\Tr ST$ is a comparison weight; it need not be a probability.

Suppose that $\supp\rho\subseteq\supp S$, and let $Q_S$ denote the support
projection of $S$. Restrict $\rho$ and $S$ to $\ran Q_S$ and replace
$T$ by $Q_STQ_S$, which is an effect on this space. There $S>0$.
This compression changes neither $\Tr\rho T$ nor $\Tr ST$. It also leaves
$E_\gamma(\rho\|S)$ unchanged: both $\rho$ and $S$ vanish on the
orthogonal complement, and so does their difference. All remaining
operator calculations in this section take place on this finite positive
support. We use Umegaki's relative entropy
\cite{Umegaki1962} and the centered relative-entropy variance
\cite[Eq.~(1)]{Li2014}:
\begin{align}
D(\rho\|S)&\coloneqq\Tr\rho(\ln\rho-\ln S), \label{eq:DV-step-1}\\
V(\rho\|S)&\coloneqq
 \Tr\rho\bigl(\ln\rho-\ln S-D(\rho\|S)I\bigr)^2 \label{eq:DV}
  = \Tr\rho(\ln\rho-\ln S)^2-D(\rho\|S)^2.
\end{align}
The logarithm of $\rho$ is taken on its positive support. As an auxiliary
Hermitian matrix it may be extended by zero on its kernel. This creates
no ambiguity in these definitions, since
\begin{equation}
 V(\rho\|S)
 =\norm{(\ln\rho-\ln S-D(\rho\|S)I)\sqrt\rho}_2^2
 \label{eq:variance-hs}
\end{equation}
and a change confined to $\ker\rho$ annihilates $\sqrt\rho$.
In particular, $V\geq0$. Note that the quantity $D$ may be negative for unnormalized
$S$. 

\subsection{Logarithmic trace inequality}
\begin{lemma}[Logarithmic trace inequality]\label{lem:log}
Let $A\geq0$, $B>0$, and $P\coloneqq\{A-B>0\}$. Then
\begin{equation}
 \Tr\!\left[PA(\ln A-\ln B)\right]\geq0.
 \label{eq:log-trace}
\end{equation}
The trace in~\eqref{eq:log-trace} is real. For singular $A$, the first
term uses $0\ln0=0$.
\end{lemma}

The exact $A$-weighted logarithmic form in Lemma~\ref{lem:log}
appears explicitly in 
\cite[Eq.~(A11)]{SharmaWarsi2012}.
Their proof builds on closely related $B$-weighted trace inequalities
for operator-monotone functions; see Audenaert et al.
\cite[Lemma~1]{Audenaert2007} for power functions and Petz
\cite[Theorem~11.18]{Petz2008} for the operator-monotone formulation.
Sharma and Warsi's quantum Chebyshev inequality appears in
\cite[Lemma~6, Eq.~(A12)]{SharmaWarsi2012}.
Cheng and Hsieh restate and use it in
\cite[Lemma~20]{ChengHsieh2018}, with the
hockey-stick converse calculation in
\cite[Eqs.~(A.61)--(A.68)]{ChengHsieh2018}.
Sharma and Warsi's nonnegative-projector convention gives the strict convention
by applying the inequality to $A$ and $B+\varepsilon I$ and letting
$\varepsilon\downarrow0$. In the finite-dimensional case, the projector is then
exactly $\{A-B>0\}$ for all sufficiently small positive $\varepsilon$.
For convenience, Appendix~\ref{app:trace} supplies a direct proof, including singular~$A$.

The projector $\{\rho-\gamma S>0\}$ need not equal
$\{\ln\rho-\ln S>\ln\gamma\}$ when $\rho$ and $S$ do not commute. Lemma~\ref{lem:log} is the
noncommutative replacement for making that identification. It bounds a
weighted logarithmic average on the actual Neyman--Pearson subspace.

\begin{proposition}[Centered quantum Chebyshev and hockey-stick bound]\label{prop:cheb}
Let $\rho$ be a state and $S\geq0$ an operator on a finite-dimensional
Hilbert space, with $\supp\rho\subseteq\supp S$.
Define $D\coloneqq D(\rho\|S)$ and $V\coloneqq V(\rho\|S)$ by
\eqref{eq:DV-step-1}--\eqref{eq:DV}.
For every $t>0$, set $P_t\coloneqq\{\rho-e^{D+t}S>0\}$. Then
\begin{equation}
 E_{e^{D+t}}(\rho\|S)
 \leq\Tr\rho P_t\leq\frac{V}{V+t^2}.
 \label{eq:cheb}
\end{equation}
Consequently, for every effect $T$,
\begin{equation}
 \Tr\rho T\leq\frac{V}{V+t^2}+e^{D+t}\Tr ST.
 \label{eq:cheb-test}
\end{equation}
\end{proposition}

\begin{proof}
We work on $\supp S$. Set $P\coloneqq P_t$, $f\coloneqq\Tr\rho P$, and
$K\coloneqq\ln\rho-\ln S-DI$. Then $\Tr\rho K=0$ and $\Tr\rho K^2=V$.
Apply Lemma~\ref{lem:log} to $\rho$ and $e^{D+t}S$.
Since $\ln(e^{D+t}S)=(D+t)I+\ln S$, it gives
\begin{equation}
 0\leq\Tr[P\rho(K-tI)]\implies tf\leq\Tr P\rho K.
 \label{eq:projector-threshold}
\end{equation}
For commuting $\rho$ and $S$, this is the usual threshold estimate for
the centered log-likelihood ratio. Lemma~\ref{lem:log} gives the same
weighted estimate for the actual Neyman--Pearson projector even when
the operators do not commute.
The last trace is real and nonnegative. Centering the projector does not
change it, since
\begin{equation}
 \Tr(P-fI)\rho K=\Tr P\rho K-f\Tr\rho K=\Tr P\rho K.
 \label{eq:centered-projector}
\end{equation}
For the Hilbert--Schmidt inner product, cyclicity of trace implies that
\begin{equation}
 \langle K\sqrt\rho,(P-fI)\sqrt\rho\rangle_{\rm HS}
 =\Tr\rho K(P-fI)=\Tr(P-fI)\rho K.
 \label{eq:centered-inner-product}
\end{equation}
Cauchy--Schwarz and $P^2=P$ therefore imply that
\begin{align}
t^2f^2
 &\leq\abs{\Tr(P-fI)\rho K}^2 \label{eq:centered-cauchy-1}\\
&\leq\norm{K\sqrt\rho}_2^2\norm{(P-fI)\sqrt\rho}_2^2 \label{eq:centered-cauchy-2}\\
&=V\Tr\rho(P-fI)^2 \label{eq:centered-cauchy-3}\\
&=V\bigl(\Tr\rho P-2f\Tr\rho P+f^2\Tr\rho\bigr) \label{eq:centered-cauchy-4}\\
&=Vf(1-f). \label{eq:centered-cauchy-5}
\end{align}
If $f=0$, the probability bound holds. Otherwise divide by $f$ to obtain
$t^2f\leq V(1-f)$, or $(V+t^2)f\leq V$. Since $t>0$, division by
$V+t^2$ is valid, including when $V=0$. This proves the second inequality
in~\eqref{eq:cheb}. Its first inequality follows from
\begin{equation}
 E_{e^{D+t}}(\rho\|S)
 =\Tr P\rho-e^{D+t}\Tr PS\leq\Tr P\rho.
 \label{eq:positive-part-acceptance}
\end{equation}
Finally, apply~\eqref{eq:hockey-test} with $\gamma=e^{D+t}$ to obtain
\eqref{eq:cheb-test}.
\end{proof}

Let us note here that the general variance-to-converse strategy
was already used in \cite{SharmaWarsi2012,ChengHsieh2018}.
In relation to Proposition~\ref{prop:cheb}, the cited versions of
quantum Chebyshev feature $V/t^2$ in place of
$V/(V+t^2)$. The latter is the elementary centered-projector refinement
proved above. Replacing
it by the established $V/t^2$ bound throughout would change a denominator
in the final finite-blocklength bound but not the strong-converse
conclusion.

\section{A uniform information-variance bound for balanced loss}\label{sec:variance}

\begin{lemma}[Balanced-channel variance]\label{lem:variance}
Let $R$ be finite dimensional, and let $\ket{\phi}_{RA^n}$ be a pure
state with finite total-photon-number support. Define
\begin{equation}
 \ket{\psi}_{RB^nE^n}\coloneqq(I_R\otimes V_{1/2}^{\otimes n})\ket{\phi},
 \qquad\omega_{RB^n}\coloneqq\Tr_{E^n}\proj{\psi},
 \qquad\sigma_{B^n}\coloneqq\Tr_R\omega_{RB^n}.
 \label{eq:balanced-state}
\end{equation}
Then, with natural logarithms,
\begin{equation}
 D(\omega_{RB^n}\|I_R\otimes\sigma_{B^n})=0,
 \qquad
 V(\omega_{RB^n}\|I_R\otimes\sigma_{B^n})\leq4n.
 \label{eq:balanced-variance}
\end{equation}
The bound is independent of the photon cutoff, the reference dimension,
and correlations among all input modes.
\end{lemma}

The proof of Lemma~\ref{lem:variance} is the channel-specific analytical
part of our paper and is given in Sections~\ref{sec:log-fluc}--\ref{sec:complete-variance-est}.
We first express the variance as a difference of
logarithms on the receiver and environment. We then explain how the
vacuum in an orthogonal mode combination bounds this difference.

\subsection{The logarithmic fluctuation on the two physical outputs}
\label{sec:log-fluc}
The coefficients of~\eqref{eq:isometry} at $t=1/2$ are invariant under
interchanging $k$ and $m-k$. Therefore the balanced isometry is fixed by
receiver--environment exchange, including on inputs correlated with $R$
and across modes. Identifying the output mode spaces, we have
\begin{equation}
 \psi_{B^n}=\psi_{E^n}=\sigma.
 \label{eq:equal-marginals}
\end{equation}
Every state in this calculation has finite-dimensional support. The
support inclusion
$\supp\omega_{RB^n}\subseteq R\otimes\supp\sigma$ follows from
positivity: if $P_\sigma$ is the support projection of $\sigma$, then
\begin{equation}
 \Tr\omega[I_R\otimes(I-P_\sigma)]
 =\Tr\sigma(I-P_\sigma)=0.
 \label{eq:marginal-support}
\end{equation}
Writing $Q_\perp\coloneqq I_R\otimes(I-P_\sigma)$, the preceding trace equals
$\norm{Q_\perp\sqrt\omega}_2^2$. Thus $Q_\perp\sqrt\omega=0$,
which implies $Q_\perp\omega=\omega Q_\perp=0$ and proves the inclusion.

Because $\ket{\psi}$ is pure, its $RB^n$ and $E^n$ marginals have the same
nonzero eigenvalues. Thus
\begin{align}
D(\omega\|I_R\otimes\sigma)
 &=\Tr\omega\ln\omega-\Tr\sigma\ln\sigma \label{eq:balanced-mean-1}\\
&=-H(RB^n)_\psi+H(B^n)_\psi \label{eq:balanced-mean-2}\\
&=-H(E^n)_\psi+H(B^n)_\psi=0. \label{eq:balanced-mean-3}
\end{align}
Here $H(C)_\psi\coloneqq H(\psi_C)$, with the von Neumann entropy defined in~\eqref{eq:vNeu-ent}.
The difference $H(B^n)_\psi-H(RB^n)_\psi$ is the coherent information
from the reference to the receiver, or equivalently minus the conditional
entropy $H(R|B^n)_\psi$.
The equality of entropies here uses only purity and~\eqref{eq:equal-marginals}.

Write a Schmidt decomposition across $RB^n:E^n$ as
$\ket{\psi}=\sum_j\sqrt{\lambda_j}\,\ket{u_j}_{RB^n}\otimes\ket{v_j}_{E^n}$, using only positive
Schmidt coefficients. Applying either $\ln\omega$ on $RB^n$ or
$\ln\sigma$ on $E^n$ multiplies its $j$th term by $\ln\lambda_j$.
Thus
\begin{equation}
 (\ln\omega\otimes I_{E^n})\ket{\psi}
 =(I_{RB^n}\otimes\ln\sigma_{E^n})\ket{\psi}.
 \label{eq:schmidt-logarithm}
\end{equation}
Define the self-adjoint auxiliary operator
\begin{equation}
 Y\coloneqq\ln\sigma_{B^n}-\ln\sigma_{E^n},
 \label{eq:Y}
\end{equation}
where each logarithm is taken on the positive support and extended by
zero on its kernel. It follows that
\begin{equation}
 V(\omega\|I_R\otimes\sigma)=\norm{Y\ket{\psi}}^2.
 \label{eq:variance-vector}
\end{equation}
The two terms defining $Y$ commute because they act on different
physical outputs. No commutation of $\omega$ with $I_R\otimes\sigma$
is asserted or used.

\subsection{Bright--dark coordinates and the photon-number gap}

Our next goal is to bound $\norm{Y\ket{\psi}}^2$. To do so, we identify a
combination of the output modes that remains in the vacuum, regardless
of the input state. This combination provides a photon-number operator
with a known spectral gap. We first define the required mode coordinates
and then apply that gap to $Y\ket{\psi}$.

For each $i\in\{1,\ldots,n\}$, let $b_i$ and $e_i$ be the annihilation
operators of the receiver and environment modes. Define
\begin{equation}
 a_{{\rm br},i}\coloneqq\frac{b_i+e_i}{\sqrt2},
 \qquad d_i\coloneqq\frac{b_i-e_i}{\sqrt2}.
 \label{eq:modes}
\end{equation}
We call these the \emph{bright} and \emph{dark} modes, respectively,
following standard terminology in quantum optics and multimode systems
\cite{Benisty2009,Lahteenmaki2016}. In the present setting, the names have
a direct physical interpretation:
the bright modes contain the original signal, whereas the dark modes
contain the vacuum that entered the unused beam-splitter ports.

To verify this interpretation, let $v_i$ denote the annihilation operator
at the unused input port, and let $U_i\colon A_iV_i\to B_iE_i$ be the
two-mode beam-splitter unitary. The real orthogonal beam-splitter matrix
consistent with~\eqref{eq:isometry} gives the Heisenberg relations
\begin{equation}
 U_i^\dagger b_iU_i=\frac{a_i+v_i}{\sqrt2},
 \qquad U_i^\dagger e_iU_i=\frac{a_i-v_i}{\sqrt2}.
 \label{eq:balanced-two-input-transform}
\end{equation}
Conjugating~\eqref{eq:modes} by $U_i$ therefore gives
$U_i^\dagger a_{{\rm br},i}U_i=a_i$ and $U_i^\dagger d_iU_i=v_i$.
Thus the change from physical output coordinates to bright
and dark coordinates is the inverse beam-splitter transformation.
It is a change of coordinates on the joint output, used to analyze the
state; it requires no additional operation in the communication protocol.

Under the corresponding Fock-space identification, the state has the form
\begin{equation}
 \ket{\psi}_{R A_{\rm br}^nD^n}
 =\ket{\phi}_{R A_{\rm br}^n}\otimes\ket{0,\ldots,0}_{D^n}.
 \label{eq:bright-dark-state-factorization}
\end{equation}
The input system $A^n$ has been identified with the bright system
$A_{\rm br}^n$. Eq.~\eqref{eq:bright-dark-state-factorization}
holds also when the input modes are entangled with each other and with
$R$. In particular,
\begin{equation}
 d_i\ket{\psi}=0,
 \qquad b_i\ket{\psi}=e_i\ket{\psi}
 \quad\text{for every }i.
 \label{eq:vacuum-relation}
\end{equation}

Let $P_0$ project onto the joint dark-mode vacuum, with the bright modes
unrestricted, and let $N_{\rm dk}$ count the total number of dark photons:
\begin{equation}
 P_0\coloneqq I_{A_{\rm br}^n}\otimes\proj{0,\ldots,0}_{D^n},
 \qquad N_{\rm dk}\coloneqq\sum_{i=1}^n d_i^\dagger d_i.
 \label{eq:dark-number}
\end{equation}
Identity operators on $R$ are implicit. Then $P_0\ket{\psi}=\ket{\psi}$.
On a dark-mode Fock vector,
\begin{equation}
 N_{\rm dk}\ket{m_1,\ldots,m_n}
 =\left(\sum_{i=1}^n m_i\right)\ket{m_1,\ldots,m_n}.
 \label{eq:dark-number-spectrum}
\end{equation}
The eigenvalue of $N_{\rm dk}$ is equal to zero on the joint vacuum and at least one on the
orthogonal complement. Consequently, the following operator inequality holds:
\begin{equation}
 N_{\rm dk}\geq I-P_0.
 \label{eq:number-gap}
\end{equation}
This operator inequality will be used in \eqref{eq:gap-gradient-step-1} below.

To apply~\eqref{eq:number-gap}, we show that $Y\ket{\psi}$ has no dark-vacuum
component. Let $S_{BE}$ be the unitary that exchanges all receiver and
environment modes. In the coordinates in~\eqref{eq:modes}, it acts as
\begin{equation}
 S_{BE}a_{{\rm br},i}S_{BE}^\dagger=a_{{\rm br},i},
 \qquad S_{BE}d_iS_{BE}^\dagger=-d_i.
 \label{eq:swap-bright-dark}
\end{equation}
It fixes the joint vacuum as well. Since receiver--environment exchange leaves each bright-mode creation
operator invariant and sends each dark-mode creation operator to its
negative, every dark-mode creation operator contributes a factor of $-1$
under the swap. Consequently, a bright--dark Fock state containing a
total of $m$ dark photons acquires the phase $(-1)^m$. Since
$N_{\rm dk}$ counts the total number of dark photons, receiver--environment
exchange realizes the dark-mode photon-number parity operator:
\begin{equation}
 S_{BE}=(-1)^{N_{\rm dk}}.
 \label{eq:swap-dark-parity}
\end{equation}
A direct Fock-basis verification of~\eqref{eq:swap-dark-parity} is given
in Appendix~\ref{app:dark-parity}.
We then conclude that
\begin{equation}
 S_{BE}P_0=P_0S_{BE}=P_0.
 \label{eq:swap-P0}
\end{equation}
The receiver and environment marginals are equal by
\eqref{eq:equal-marginals}, so exchange interchanges the two terms in
$Y=\ln\sigma_{B^n}-\ln\sigma_{E^n}$. Hence
\begin{equation}
 S_{BE}YS_{BE}^\dagger=-Y.
 \label{eq:Y-odd}
\end{equation}
Combining~\eqref{eq:swap-P0} and~\eqref{eq:Y-odd} gives
\begin{equation}
 P_0YP_0=P_0S_{BE}YS_{BE}^\dagger P_0=-P_0YP_0,
 \quad\text{and therefore}\quad P_0YP_0=0.
 \label{eq:P0YP0-zero}
\end{equation}
Since $P_0\ket{\psi}=\ket{\psi}$, it follows that
\begin{equation}
 P_0Y\ket{\psi}=0.
 \label{eq:Ypsi-orthogonal}
\end{equation}
Equivalently, $Y$ takes the even-parity vector $\ket{\psi}$ into the odd
dark-parity sector, which is orthogonal to the dark vacuum.

We now apply the operator inequality in \eqref{eq:number-gap}:
\begin{align}
 \norm{Y\ket{\psi}}^2
 &=\bra{\psi}Y(I-P_0)Y\ket{\psi}
 \label{eq:gap-gradient-step-0}\\
 &\leq\bra{\psi}YN_{\rm dk}Y\ket{\psi}
 \label{eq:gap-gradient-step-1}\\
 &=\sum_{i=1}^n\norm{d_iY\ket{\psi}}^2
 \label{eq:gap-gradient-step-4}\\
 &=\sum_{i=1}^n\norm{[d_i,Y]\ket{\psi}}^2.
 \label{eq:gap-gradient}
\end{align}
The first equality uses~\eqref{eq:Ypsi-orthogonal}; the inequality uses
\eqref{eq:number-gap}; the next equality uses the definition of
$N_{\rm dk}$; and the last uses $d_i\ket{\psi}=0$ from
\eqref{eq:vacuum-relation}. Together with~\eqref{eq:variance-vector},
this proves
\begin{equation}
 V(\omega\|I_R\otimes\sigma)
 \leq\sum_{i=1}^n\norm{[d_i,Y]\ket{\psi}}^2.
 \label{eq:variance-gradient-bound}
\end{equation}
Thus we have bounded the logarithmic fluctuation by a sum
of squared commutator norms. It remains to bound
each summand independently of the input energy. The next subsection
obtains that bound by creating one photon in a dark mode.

All of these quadratic forms are well-defined under the finite-support
assumption. If the input has total photon number at most $L$, then each
reduced state is supported on its local $L$-photon cutoff. Its logarithm,
extended by zero on its kernel, maps into that finite-dimensional space.
Applying either logarithm to $\ket{\psi}$ therefore produces a vector with
total photon number at most $2L$ on $B^nE^n$. Passive changes of mode
coordinates preserve total photon number, so $Y\ket{\psi}$ belongs to the
domains of every $d_i$ and of $N_{\rm dk}$. The proof uses the physical
Fock-space operators throughout.

\subsection{A one-photon estimate for the logarithmic gradient}
The commutators in~\eqref{eq:variance-gradient-bound} involve
logarithms of the reduced state. We first control commutators with the
state itself and then compare their matrix elements with those of the
logarithm. The following partial-trace inequality supplies the first step.

\begin{lemma}[A cross partial trace]\label{lem:cross}
Let $\ket{u}$ be a unit vector on $B\otimes Z$, let $\ket{v}$ be another
vector, and suppose $\xi_B\coloneqq\Tr_Z[\proj{u}]$ has finite rank. Set
$C\coloneqq\Tr_Z[\ketbra{u}{v}]$. If $\xi_B^+$ denotes the inverse on its positive
support and zero on its kernel, then
\begin{equation}
 \ran C\subseteq\supp\xi_B,
 \qquad \Tr C^\dagger\xi_B^+C\leq\norm{\ket{v}}^2.
 \label{eq:cross}
\end{equation}
\end{lemma}
\begin{proof}
Schmidt-expand $\ket{u}=\sum_j\sqrt{\lambda_j}\,\ket{u_j}\otimes\ket{z_j}$ with
$\lambda_j>0$ and orthonormal sets $\{\ket{u_j}\}$ and $\{\ket{z_j}\}$.
Define $\ket{w_j}\coloneqq(I_B\otimes\bra{z_j})\ket{v}$. By the definition of partial trace,
\begin{equation}
 C=\sum_j\sqrt{\lambda_j}\ketbra{u_j}{w_j},
 \qquad
 C^\dagger\xi_B^+C=\sum_j\proj{w_j}.
 \label{eq:schmidt-cross}
\end{equation}
The first identity proves the range inclusion. The second gives
\begin{equation}
 \Tr C^\dagger\xi_B^+C=\sum_j\norm{\ket{w_j}}^2
 =\bra{v}I_B\otimes\Bigl(\sum_j\proj{z_j}\Bigr)\ket{v}
 \leq\norm{\ket{v}}^2,
 \label{eq:cross-trace}
\end{equation}
since the sum in parentheses is an orthogonal projection. This proves
\eqref{eq:cross}.
\end{proof}

Return to the balanced-channel vector $\ket{\psi}_{RB^nE^n}$ and its
receiver marginal $\sigma_{B^n}$ defined in~\eqref{eq:balanced-state}.
We continue to write $\sigma$ for this marginal. For each $i$, create one
photon in the dark mode:
\begin{equation}
 \ket{\chi_i}\coloneqq d_i^\dagger\ket{\psi}.
 \label{eq:one-photon-vector}
\end{equation}
The canonical commutation relation and~\eqref{eq:vacuum-relation} give
\begin{equation}
 \norm{\ket{\chi_i}}^2
 =\bra{\psi}d_id_i^\dagger\ket{\psi}
 =1+\bra{\psi}d_i^\dagger d_i\ket{\psi}=1.
 \label{eq:one-photon-norm}
\end{equation}
These vectors have total photon number at most $L+1$. Set
$C_i\coloneqq\Tr_{RE^n}\ketbra{\psi}{\chi_i}$. Since
$\bra{\chi_i}=\bra{\psi} d_i$, partial-trace cyclicity on the traced
mode $E_i$ gives
\begin{align}
C_i
 &=\frac1{\sqrt2}\left(\sigma b_i-
       \Tr_{RE^n}[\proj{\psi}\, e_i]\right) \label{eq:one-photon-cross-1}\\
&=\frac1{\sqrt2}\left(\sigma b_i-
       \Tr_{RE^n}[e_i\proj{\psi}]\right) \label{eq:one-photon-cross-2}\\
&=\frac{\sigma b_i-b_i\sigma}{\sqrt2}. \label{eq:one-photon-cross-3}
\end{align}
The last equality uses $e_i\ket{\psi}=b_i\ket{\psi}$ from~\eqref{eq:vacuum-relation}.
Every product in this identity is finite rank and acts on a finite
photon-number subspace enlarged by one photon if necessary. The cyclicity
step therefore reduces to a finite matrix identity.

Lemma~\ref{lem:cross}, with the vectors $\ket{\psi}$ and $\ket{\chi_i}$ and traced system
$RE^n$, now implies
\begin{equation}
 \ran[b_i,\sigma]\subseteq\supp\sigma,
 \qquad
 \Tr[b_i,\sigma]^\dagger\sigma^+[b_i,\sigma]\leq2.
 \label{eq:inverse-gradient}
\end{equation}
The inverse in~\eqref{eq:inverse-gradient} acts on the range of
$[b_i,\sigma]$. This placement will match the denominator in the
eigenvalue comparison below.

We next justify the use of logarithms. Let $P_\sigma$ denote the support
projection of $\sigma$. The range statement in~\eqref{eq:inverse-gradient} gives
$(I-P_\sigma)[b_i,\sigma]=0$. Since $(I-P_\sigma)\sigma=0$, this means
$(I-P_\sigma)b_i\sigma=0$. Multiplying on the right by $\sigma^+$ gives
\begin{equation}
 (I-P_\sigma)b_iP_\sigma=0.
 \label{eq:support-invariance}
\end{equation}
Thus lowering cannot take a vector in the positive support of $\sigma$
into its kernel.

For $r>0$, the scalar inequality
\begin{equation}
 r(\ln r)^2\leq(r-1)^2
 \label{eq:scalar-log-bound}
\end{equation}
is verified by setting $x\coloneqq\tfrac12\ln r$, so that $r=e^{2x}$. After division by $4e^{2x}$, the inequality in \eqref{eq:scalar-log-bound} is equivalent to
$x^2\leq\sinh^2x$. For $x\geq0$, $\sinh x=\int_0^x\cosh y\,dy\geq x$;
oddness gives the result for $x<0$. Therefore, for strictly positive eigenvalues
$\lambda_j,\lambda_k$ of $\sigma$, setting $r=\frac{\lambda_j}{\lambda_k}$ in \eqref{eq:scalar-log-bound} leads to
\begin{equation}
 \lambda_j(\ln\lambda_j-\ln\lambda_k)^2
 \leq\frac{(\lambda_j-\lambda_k)^2}{\lambda_k}.
 \label{eq:scalar-eigenvalue-bound}
\end{equation}
For the following matrix sums, choose an orthonormal basis $\{\ket{u_j}\}$ of
the subspace with at most $L+1$ photons on $B^n$ such that
$\sigma\ket{u_j}=\lambda_j\ket{u_j}$, including the zero eigenvalues.
Both $b_i\sigma$ and $\sigma b_i$ vanish outside
this finite input and output window. The additional zero-eigenvalue
indices are important: $\sigma b_i$ can be nonzero on the $(L+1)$-photon
shell even though $\sigma$ itself has support below that shell.
Applying~\eqref{eq:support-invariance}, we obtain
\begin{equation}
 \begin{aligned}
 \Tr\sigma[b_i,\ln\sigma]^\dagger[b_i,\ln\sigma]
 &=
 \sum_{j,k:\lambda_j,\lambda_k>0}
 \lambda_j\abs{\bra{u_k}b_i\ket{u_j}}^2
 (\ln\lambda_j-\ln\lambda_k)^2\\
 &\quad\leq
 \sum_{j,k:\lambda_j,\lambda_k>0}
 \frac{\abs{\bra{u_k}b_i\ket{u_j}}^2
       (\lambda_j-\lambda_k)^2}{\lambda_k}\\
 &\quad\leq
 \sum_{k:\lambda_k>0}\sum_j
 \frac{\abs{\bra{u_k}b_i\ket{u_j}}^2
       (\lambda_j-\lambda_k)^2}{\lambda_k}\\
 &\quad=\Tr[b_i,\sigma]^\dagger\sigma^+[b_i,\sigma]\leq2.
 \end{aligned}
 \label{eq:log-gradient-chain}
\end{equation}
In the first sum, terms with $\lambda_j>0$ and $\lambda_k=0$ vanish by
\eqref{eq:support-invariance}; terms with $\lambda_j=0$ have zero weight.
The second inequality adds only nonnegative terms. We have proved
\begin{equation}
 \Tr\sigma[b_i,\ln\sigma]^\dagger[b_i,\ln\sigma]\leq2.
 \label{eq:log-gradient}
\end{equation}
This argument also shows that the weighted expression is unaffected by
the extension of $\ln\sigma$ to its kernel.

\subsection{Completion of the variance estimate}
\label{sec:complete-variance-est}
From~\eqref{eq:Y} and~\eqref{eq:modes}, operators acting on different output
systems commute, and hence
\begin{equation}
 [d_i,Y]=\frac{[b_i,\ln\sigma_{B^n}]+[e_i,\ln\sigma_{E^n}]}{\sqrt2}.
 \label{eq:dark-log-commutator}
\end{equation}
For any vectors $\ket{v}$ and $\ket{w}$, expanding the squared norm and using
$2\operatorname{Re}\langle v|w\rangle\leq\norm{\ket{v}}^2+\norm{\ket{w}}^2$ proves
$\norm{(\ket{v}+\ket{w})/\sqrt2}^2\leq\norm{\ket{v}}^2+\norm{\ket{w}}^2$.
Consequently, by~\eqref{eq:equal-marginals} and~\eqref{eq:log-gradient},
\begin{align}
\norm{[d_i,Y]\ket{\psi}}^2
 &\leq\norm{[b_i,\ln\sigma_{B^n}]\ket{\psi}}^2
        +\norm{[e_i,\ln\sigma_{E^n}]\ket{\psi}}^2 \label{eq:commutator-sum-1}\\
&=2\Tr\sigma[b_i,\ln\sigma]^\dagger[b_i,\ln\sigma]\leq4. \label{eq:commutator-sum-2}
\end{align}
Summing in~\eqref{eq:gap-gradient} and using~\eqref{eq:variance-vector}
gives $V(\omega\|I_R\otimes\sigma)\leq4n$. The zero-mean identity was
already proved, completing the proof of Lemma~\ref{lem:variance}.

\begin{corollary}[Testing a balanced-channel output]\label{cor:balanced}
For the state in Lemma~\ref{lem:variance}, every effect $T$ on $RB^n$ and
every $a>0$ satisfy
\begin{equation}
 \Tr\omega T\leq\frac{4n}{4n+a^2}
       +e^a\Tr[(I_R\otimes\sigma)T].
 \label{eq:balanced-test}
\end{equation}
\end{corollary}
\begin{proof}
Apply~\eqref{eq:cheb-test} with $S\coloneqq I_R\otimes\sigma$, for which $D=0$, and evaluate it at $t=a$.
The function $x\mapsto x/(x+a^2)$ is increasing for $x\geq0$, because
its derivative is $a^2/(x+a^2)^2$. Lemma~\ref{lem:variance} therefore
bounds the first term by $4n/(4n+a^2)$, thus proving~\eqref{eq:balanced-test}.
\end{proof}

\section{The decoder comparison trace at general transmissivity}\label{sec:transport}

\subsection{Two factorizations of the same Gaussian isometry}
Corollary~\ref{cor:balanced} supplies a test bound for balanced loss.
We now apply the factorization from Section~\ref{sec:balanced-overview} to a general pure-loss
channel and compute the comparison weight of the decoder test.
Fix $1/2\leq\eta<1$ and set
\begin{equation}
 c\coloneqq1-\eta,\qquad \tau\coloneqq\frac c\eta\in(0,1],
 \qquad \nu\coloneqq2\eta-1\geq0.
 \label{eq:parameters}
\end{equation}
With the parameters in~\eqref{eq:parameters}, let
$W\coloneqq V_\tau^{\otimes n}$, with $W\colon B^n\to E'^nF^n$, where $E'$ receives the
transmitted fraction $\tau$. The channel obtained by tracing out $F^n$
is the degrading attenuator identified in~\cite[Fig.~1]{GuhaShapiroErkmen2008}.
The composed isometry, one mode at a time,
acts on the input creation operator $a^\dagger$ as
\begin{equation}
 a^\dagger\longmapsto
 \sqrt c\,e^\dagger+\sqrt c\,e'^\dagger+\sqrt \nu\,f^\dagger.
 \label{eq:three-arms}
\end{equation}
Indeed, the original receiver amplitude $\sqrt\eta$ is split into
$\sqrt{\eta\tau}=\sqrt c$ and
$\sqrt{\eta(1-\tau)}=\sqrt \nu$.
The same map is obtained by first splitting $A$ into $G$ and $F$ with fractions
$2c$ and $\nu$, respectively, and then applying a balanced splitter to $G$:
\begin{equation}
 a^\dagger\longmapsto\sqrt{2c}\,g^\dagger+\sqrt \nu\,f^\dagger,
 \qquad g^\dagger\longmapsto(e^\dagger+e'^\dagger)/\sqrt2.
 \label{eq:balanced-factorization}
\end{equation}
See Figure~\ref{fig:balanced-decomposition}. This gives the bosonic
realization of Morgan and Winter's symmetric-channel construction
discussed in 
\cite[Section~VII]{MorganWinter2014}. The isometry on $G$ is exactly
$V_{1/2}\colon G\to EE'$, and exchanging $E$ and $E'$ leaves it unchanged.
Tracing either output therefore gives the same balanced pure-loss
channel. The factorization is established below directly on Fock space;
we do not apply their finite-dimensional strong-converse reduction
theorem to the bosonic channel.

For a direct check at the level of isometries, let $u,v,w\geq0$ with
$u+v+w=m$. In the original channel followed by $W$, the coefficient of
$\ket u_E\ket v_{E'}\ket w_F$ is
\begin{equation}
 \sqrt{\binom m u\binom{m-u}v}\,
 c^{u/2}\eta^{(v+w)/2}\tau^{v/2}(1-\tau)^{w/2}
 =\sqrt{\frac{m!}{u!v!w!}}\,c^{(u+v)/2}\nu^{w/2}.
 \label{eq:first-factorization-coefficient}
\end{equation}
In the other ordering it is
\begin{equation}
 \sqrt{\binom m w\binom{m-w}u}\,
 (2c)^{(u+v)/2}\nu^{w/2}2^{-u/2}2^{-v/2}
=\sqrt{\frac{m!}{u!v!w!}}\,c^{(u+v)/2}\nu^{w/2}.
 \label{eq:second-factorization-coefficient}
\end{equation}
The coefficients agree on every Fock vector. Equality extends to the
full isometries by linearity and continuity, and tensoring preserves
it for correlated multimode inputs. At $\eta=1/2$, only the $w=0$ terms
remain: $F$ is a fixed vacuum mode and the first splitter is the identity
on $G$. Thus this endpoint is included.

For a pure encoded input, let
\begin{equation}
 \ket{\psi}_{RB^nE^n}\coloneqq(I_R\otimes V_\eta^{\otimes n})\ket{\phi}_{RA^n},
 \qquad
 \ket{\zeta}_{RE^nE'^nF^n}\coloneqq(I_{RE^n}\otimes W)\ket{\psi}.
 \label{eq:three-output-state}
\end{equation}
We use $\psi$ and $\zeta$ with subsystem subscripts for the corresponding
density operators and their marginals. By~\eqref{eq:three-arms}, $\zeta$ is an actual
balanced-channel output with composite reference $RF^n$, receiver $E'^n$, and
environment $E^n$. The state before that balanced splitter is pure on
$RF^nG^n$. If the encoded input has finite photon support, the reference
$RF^n$ may be restricted to a finite-dimensional space containing its
support. Corollary~\ref{cor:balanced} applies to every effect on this
reference and $E'^n$.

\subsection{The attenuator identity on bounded operators}

The identity $\cL_\tau(I)=\tau^{-1}I$ was established in the operator-sum
analysis of Ivan, Sabapathy, and Simon~\cite[Eq.~(5.6), arXiv version]{IvanSabapathySimon2011}.
Their attenuation amplitude $\kappa'$ satisfies $\tau=(\kappa')^2$,
so their factor $(\kappa')^{-2}$ is $\tau^{-1}$ in our transmissivity
convention.

We include a derivation here in order to explain how this identity is used, despite
the fact that the identity operator on Fock space is not trace class.
The loss Kraus operators for transmissivity~$\tau$ are given by
\begin{equation}
 K_\ell\ket m\coloneqq\begin{cases}
 \sqrt{\binom m\ell}\,\tau^{(m-\ell)/2}(1-\tau)^{\ell/2}
 \ket{m-\ell},&m\geq\ell,\\
 0,&m<\ell.
 \end{cases}
 \label{eq:loss-kraus}
\end{equation}
This follows by taking $\bra\ell_F$ in~\eqref{eq:isometry} for $V_\tau$.
The operators $K_\ell K_\ell^\dagger$ are diagonal. For each output
number $k$, their diagonal entries sum to
\begin{align}
\sum_{\ell=0}^\infty\langle k|K_\ell K_\ell^\dagger|k\rangle
 &=\sum_{\ell=0}^\infty\binom{k+\ell}{\ell}\tau^k(1-\tau)^\ell \label{eq:identity-kraus-1}\\
&=\tau^k\tau^{-k-1}=\tau^{-1}. \label{eq:identity-kraus-2}
\end{align}
The generating function used on the second line follows by taking $k$
derivatives of $\sum_{j\geq0}x^j=(1-x)^{-1}$ and dividing by $k!$, i.e.,
\begin{equation}
 \sum_{\ell\geq0}\binom{k+\ell}{\ell}x^\ell=(1-x)^{-k-1},
 \qquad |x|<1.
 \label{eq:negative-binomial}
\end{equation}
The increasing positive diagonal sums are bounded by $\tau^{-1}I$ and
converge strongly. Consequently the Schr\"odinger-picture Kraus map has
the bounded value
\begin{equation}
 \cL_\tau^{\otimes n}(I^{\otimes n})=\tau^{-n}I^{\otimes n}.
 \label{eq:identity}
\end{equation}
For a bounded positive operator $Z$, its partial Kraus sums increase
and are bounded from above by $\norm Z_\infty\tau^{-1}I$, so they converge
strongly to a bounded operator. The resulting positive map is normal:
for an increasing bounded family of positive inputs, monotone convergence
allows the input limit and the nonnegative Kraus sum to be interchanged
in every vector expectation. This defines the extension used in
\eqref{eq:identity}. This is different
from the identity $\cL_\tau^\dag(I)=I$ for the Heisenberg adjoint.

For every finite-rank
input projection $\Pi$, positivity and $\Pi\leq I$ give
\begin{equation}
 \Tr_{F^n}[W\Pi W^\dagger]
 =\cL_\tau^{\otimes n}(\Pi)\leq\tau^{-n}I_{E'^n}.
 \label{eq:partial-range}
\end{equation}
Observe that there is no finite output-dimension term in~\eqref{eq:partial-range}.

\subsection{Evaluating the decoder test on the balanced-channel state}

Let the pure encoded state have total photon number at most $L$.
Its receiver state $\rho_{RB^n}\coloneqq\Tr_{E^n}\proj{\psi}$ is supported on
$I_R\otimes\Pi_L^B$, where $\Pi_L^B$ is~\eqref{eq:cutoff} on $B^n$.
Define
\begin{equation}
 Q\coloneqq(\id_R\otimes\cD^\dag )(\Phi_M),\qquad
 Q_L\coloneqq(I_R\otimes\Pi_L^B)Q(I_R\otimes\Pi_L^B).
 \label{eq:Q}
\end{equation}
Complete positivity and unitality of $\cD^\dag $ give $0\leq Q\leq I$.
The Bell marginal is $\Tr_R\Phi_M=I_{\widehat R}/M$, so
\begin{equation}
 \Tr_R Q=\cD^\dag (I_{\widehat R}/M)=I_{B^n}/M,
 \qquad \Tr_R Q_L=\Pi_L^B/M.
 \label{eq:Q-marginal}
\end{equation}
By the support of the receiver state and the definition of the adjoint,
the code fidelity $F$ can be expressed as
\begin{equation}
F=\Tr\rho_{RB^n}Q=\Tr\rho_{RB^n}Q_L.
\end{equation}
Only the effect is compressed in~\eqref{eq:Q}; the physical channel has
not been modified.

To evaluate the same decoder test on $RE'^nF^n$, conjugate $Q_L$ by
the degrading isometry:
\begin{equation}
 T\coloneqq(I_R\otimes W)Q_L(I_R\otimes W^\dagger),
 \qquad P_L\coloneqq W\Pi_L^BW^\dagger.
 \label{eq:T}
\end{equation}
Since $W$ is an isometry, $Q_L\leq I_R\otimes\Pi_L^B$, and
\eqref{eq:Q-marginal} holds, the effect in~\eqref{eq:T} satisfies
\begin{equation}
 0\leq T\leq I_R\otimes P_L\leq I,
 \qquad \Tr_R T=P_L/M.
 \label{eq:T-marginal}
\end{equation}
The marginal on $RE'^nF^n$ is
$(I_R\otimes W)\rho_{RB^n}(I_R\otimes W^\dagger)$.
Using $W^\dagger W=I$, isometric invariance of the trace gives
\begin{equation}
 \Tr\zeta_{RF^nE'^n}T=\Tr\rho_{RB^n}Q_L=F.
 \label{eq:F-transport}
\end{equation}
System ordering in~\eqref{eq:F-transport} only permutes tensor factors.

Set $\sigma_{E'^n}\coloneqq\zeta_{E'^n}$. For the balanced-channel state, apply
Corollary~\ref{cor:balanced} with its reference system replaced by the
composite system $RF^n$. Its reference Hilbert space may be restricted to
$\mathcal H_R\otimes\ran\Pi_L^F$, where $\Pi_L^F$ is the cutoff
\eqref{eq:cutoff} on $F^n$, since both the state and the effect $T$ have at most $L$ photons in $F^n$. The reference identity in the
testing inequality is therefore finite dimensional. It can equally be written as
$I_R\otimes I_{F^n}$ inside the following trace because $T$ has the
specified finite support. By~\eqref{eq:T-marginal} and
\eqref{eq:partial-range},
\begin{align}
\beta&\coloneqq\Tr[(I_{RF^n}\otimes\sigma_{E'^n})T] \label{eq:comparison-step-1}\\
&=\frac1M\Tr[(I_{F^n}\otimes\sigma_{E'^n})P_L] \label{eq:comparison-step-2}\\
&=\frac1M\Tr\!\left[\sigma_{E'^n}\Tr_{F^n}P_L\right] \label{eq:comparison-step-3}\\
&=\frac1M\Tr\!\left[\sigma_{E'^n}\cL_\tau^{\otimes n}(\Pi_L^B)\right] \label{eq:comparison-step-4}\\
&\leq\frac{\tau^{-n}}M\Tr\sigma_{E'^n}
 =\frac{\tau^{-n}}M. \label{eq:comparison}
\end{align}
All traces in~\eqref{eq:comparison-step-1}--\eqref{eq:comparison} are finite. If the testing lemma is
implemented after additionally compressing to $\supp\sigma_{E'^n}$,
the fidelity and comparison trace remain unchanged. The marginal
identity~\eqref{eq:T-marginal} is used before this optional additional
compression; no such identity is needed for the further compressed
effect.

The effect $T$ need not be a decoder effect for transmitting the entire
reference $RF^n$. It is simply an effect to which quantum hypothesis
testing applies. Eq.~\eqref{eq:comparison} is what accounts for
its access to $F^n$. Tracing over the original logical register gives
$P_L/M$, not an identity on all of $E'^nF^n$; the subsequent trace over
$F^n$ produces $\tau^{-n}$ rather than a Hilbert-space dimension.

Apply Corollary~\ref{cor:balanced} to the state $\zeta$ with reference $RF^n$ and
receiver $E'^n$. Eqs.~\eqref{eq:F-transport} and~\eqref{eq:comparison}
imply
\begin{equation}
 F\leq\frac{4n}{4n+a^2}+e^a\frac{\tau^{-n}}M
 =\frac{4n}{4n+a^2}
  +\exp\!\left(a+nq_\eta\ln2-\ln M\right).
 \label{eq:finite-bound}
\end{equation}
The equality uses $-\ln\tau=q_\eta\ln2$ for $\eta\geq1/2$.
Setting $a\coloneqq b\ln2$ gives~\eqref{eq:main-free} for pure inputs
with finite photon support and $1/2\leq\eta<1$.

\section{Completion of the proof for arbitrary codes}\label{sec:completion}

\subsection{Smaller transmissivities}
Vacuum attenuators obey the cascade law
\begin{equation}
 \cL_s\circ\cL_t=\cL_{st}.
 \label{eq:cascade}
\end{equation}
To verify~\eqref{eq:cascade}, use independent vacuum inputs $e_1$ and $e_2$.
After two beam splitters the receiver annihilation operator is
\begin{equation}
 b=\sqrt{st}\,a+\sqrt{s(1-t)}\,e_1+\sqrt{1-s}\,e_2.
 \label{eq:cascade-mode}
\end{equation}
The two vacuum coefficients have squared sum
$s(1-t)+(1-s)=1-st$. If $st<1$, their normalized linear combination
is a canonical vacuum mode, because an independent vacuum is invariant
under a passive change of mode coordinates. Thus $b$ is the output of
a pure-loss channel of transmissivity $st$. If $st=1$, both channels
are identities, so that the same conclusion holds.

For $0\leq\eta\leq1/2$, Eq.~\eqref{eq:cascade} gives
$\cL_\eta=\cL_{2\eta}\circ\cL_{1/2}$.
Compose the decoder with $\cL_{2\eta}^{\otimes n}$. The original code's
fidelity is then the fidelity of this composite decoder after a
balanced channel. Apply~\eqref{eq:finite-bound} at $\eta=1/2$, for which
$\tau=1$ and $q_\eta=0$. This proves~\eqref{eq:main-free} for the whole
range $0\leq\eta<1$, still for pure inputs with finite photon support.

\subsection{Removing the photon cutoff for each fixed pure input}
Fix $n$, $M$, the decoder, and an arbitrary normalized pure input
$\ket{\phi}_{RA^n}$. Let $\Pi_L^A$ be~\eqref{eq:cutoff} on the input, and set
\begin{equation}
 p_L\coloneqq\langle\phi|I_R\otimes\Pi_L^A|\phi\rangle,
 \qquad
 \ket{\phi_L}\coloneqq
 \frac{(I_R\otimes\Pi_L^A)\ket{\phi}}{\sqrt{p_L}}
 \label{eq:normalized-cutoff}
\end{equation}
whenever $p_L>0$. Strong convergence of the projections implies
$p_L\to1$, so $p_L>0$ eventually. The overlap is
$\langle\phi|\phi_L\rangle=\sqrt{p_L}$. Hence, applying the known equality relating trace distance and fidelity for pure states, we conclude that
\begin{equation}
 \frac12\norm{\proj{\phi}-\proj{\phi_L}}_1
 =\sqrt{1-\abs{\langle\phi|\phi_L\rangle}^2}
 =\sqrt{1-p_L}\longrightarrow0.
 \label{eq:pure-cutoff-distance}
\end{equation}
For this fixed code, define the corresponding input-space test operator by
\begin{equation}
 A_{\rm test}\coloneqq
 \bigl(\id_R\otimes(\cD\circ\cL_\eta^{\otimes n})^\dag\bigr)(\Phi_M).
 \label{eq:input-test}
\end{equation}
Since the adjoint is unital and completely positive,  $0\leq A_{\rm test}\leq I$.
Write $F_L\coloneqq\Tr[A_{\rm test}\proj{\phi_L}]$ for the fidelity
of the truncated input, with the same decoder.
If $Z\coloneqq\proj{\phi_L}-\proj{\phi}$, then $\Tr Z=0$, and its positive and
negative parts each have trace $\norm Z_1/2$. Thus
$-\Tr Z_-\leq\Tr A_{\rm test}Z\leq\Tr Z_+$, giving
\begin{equation}
 |F_L-F|\leq\tfrac12\norm{\proj{\phi_L}-\proj{\phi}}_1
 =\sqrt{1-p_L}\longrightarrow0.
 \label{eq:fidelity-cutoff-limit}
\end{equation}
This proves convergence using only a bounded effect, without a continuity
claim for an unbounded logarithm.

The right side of~\eqref{eq:main-free} is independent of $L$.
Consequently its inequality passes to the limit for the fixed code and
every fixed $b>0$. This proves~\eqref{eq:main-free} for arbitrary 
pure input states, including inputs with infinite mean photon number. There
is no need to pass a logarithm, variance, or entropy through this limit.
The cutoff is removed before considering an asymptotic sequence of
blocklengths.

\subsection{Mixed inputs and arbitrary encoding channels}
A density operator on the separable space $R\otimes A^n$ is
positive semidefinite and trace class, and hence has a countable spectral decomposition
\begin{equation}
 \rho_{RA^n}=\sum_j w_j\proj{\phi_j},
 \qquad w_j\geq0,\quad\sum_jw_j=1.
 \label{eq:mixed-decomposition}
\end{equation}
Every vector uses the same logical register $R$ of dimension $M$.
For a fixed decoder the fidelity is linear, and trace-norm convergence
of the positive spectral sums implies
\begin{equation}
 F(\rho)=\sum_jw_jF(\phi_j).
 \label{eq:linear-fidelity}
\end{equation}
Each term satisfies the same right side of~\eqref{eq:main-free}, so their
convex combination does too. This proves the theorem for every normal
encoded state. In particular, it covers the state generated by every
CPTP encoder acting on one half of $\Phi_M$. No purifying system has been
added to the logical register and no purity or flatness condition on
$\rho_R$ was imposed.

\subsection{Evaluation at the rate gap}
We have proved~\eqref{eq:main-free} for every code. Recall that
$\delta_n=\log_2M-nq_\eta$, as defined in~\eqref{eq:rate-gap}.
If $\delta_n>0$, set $b\coloneqq\delta_n/2$. Then
\begin{equation}
 \frac{4n}{4n+b^2(\ln2)^2}
 =\frac{16n}{16n+\delta_n^2(\ln2)^2},
 \qquad
 2^{b+nq_\eta-\log_2M}=2^{-\delta_n/2}.
 \label{eq:gap-substitution}
\end{equation}
This proves~\eqref{eq:main-gap}. If
$\log_2M\geq n(q_\eta+\Delta)$, then $\delta_n\geq n\Delta$.
Both terms on the right side of~\eqref{eq:main-gap} decrease as the
positive variable $\delta_n$ increases. Replacing it by $n\Delta$ and
dividing the first numerator and denominator by $n$ gives
\eqref{eq:main-rate}. This completes the proof of Theorem~\ref{thm:main}.

\section{Conclusion and future research}\label{sec:conclusion}

We have proved a strong converse for the unconstrained quantum capacity
of the pure-loss bosonic channel. The finite-blocklength estimate in
Theorem~\ref{thm:main} applies to every encoded state and
joint decoder, including states with infinite mean photon number.
At each fixed positive rate gap, it forces the entanglement-generation
fidelity to zero as $O(1/n)$. The result concerns unassisted quantum
communication; it does not impose Gaussianity or independence on the
encoded state.

The proof separates the capacity threshold from the fluctuations around
it. For a balanced-channel output, receiver-environment symmetry forces
the relevant relative entropy to vanish. The vacuum of the dark modes, its
photon-number gap, and a one-photon partial-trace estimate bound the
related variance by $4n$. Quantum Chebyshev testing, in hockey-stick form, converts
this uniform fluctuation estimate into an operational bound. For general
transmissivity, evaluating the decoder test on the balanced-channel state
bounds its comparison weight by $\tau^{-n}/M$. Its logarithmic
contribution is exactly the capacity term. The identity
$\cL_\tau(I)=\tau^{-1}I$ from 
\cite{IvanSabapathySimon2011} accounts for the
auxiliary bosonic output without introducing its Hilbert-space dimension.

The treatment of infinite dimension is also essential to the conclusion.
Finite photon support is used only to justify the operator calculations
for each fixed code. The resulting fidelity estimate is independent of
the support cutoff, so its removal requires convergence of a bounded
measurement probability, not continuity of an unbounded logarithm or an
interchange of a cutoff limit with a coding limit.

Several questions follow from this analysis.

\paragraph{Exponential decay above capacity.}
The present estimate proves polynomial decay of the fidelity. It does
not establish an exponential strong-converse exponent for the
unconstrained bosonic channel. One possible direction is to replace the present second-moment estimate with a stronger bound on higher or exponential moments. Such a result could lead to an exponential strong converse. The main difficulty is noncommutativity: higher-order quantities involving \(\ln\rho\) and \(\ln\sigma\) cannot in general be manipulated as though they were classical random variables. It would therefore be necessary to identify an appropriate noncommutative moment quantity and prove a bound that holds uniformly over all codes. Such control could support a
R\'enyi-based refinement of the testing argument. A variance bound alone
does not supply that exponential-moment estimate, so this step requires
additional analysis rather than a change of the free parameter in
Theorem~\ref{thm:main}.

\paragraph{Finite-blocklength behavior near capacity.}
Eq.~\eqref{eq:main-gap} already implies vanishing fidelity whenever
the positive rate gap satisfies $\delta_n/\sqrt n\to\infty$: then
$\delta_n^2/n\to\infty$ and $\delta_n\to\infty$, so both terms vanish.
It remains to determine the optimal behavior when $\delta_n$ is of order
$\sqrt n$, to obtain matching achievable bounds, and to assess whether
the constant in the uniform variance estimate can be improved.
The present argument neither identifies a dispersion coefficient nor
proves a sharp second-order expansion.

\paragraph{Energy-constrained thresholds.}
A separate question is the strong converse at a finite-energy
quantum-capacity threshold~\cite{WildeQi2018}. Such a result must specify
its code constraint carefully: a mean-photon-number constraint and a
photon-number occupation constraint are not interchangeable, as emphasized
in optical strong-converse work~\cite{WildeWinter2014,Bardhan2015}.
An energy-dependent refinement would have to retain the relevant
constraint when bounding the decoder test's comparison trace or the
Gaussian fluctuation, rather than rely solely on the energy-independent
bound $\tau^{-n}$.

\paragraph{Other Gaussian channels.}
It is natural to ask which parts of this fluctuation-based argument
survive for thermal-noise attenuators, noisy amplifiers, or more general
multimode Gaussian channels. The strong converse for the quantum-limited
amplifier is already known
\cite[Theorem~24 and Corollary~25]{WildeTomamichelBerta2017}; in that case,
the question is whether the present physical fluctuation estimates give
a useful alternative quantitative bound, not whether the strong-converse
property itself remains open. An extension to noisy channels would need
a replacement for the exact balanced-output
symmetry and the dark-vacuum estimate, together with a comparison-trace
bound appropriate to the channel. Those are model-dependent requirements;
the present proof does not establish them. More broadly, the argument
suggests examining when constraints on a channel's Stinespring range
produce code-uniform logarithmic fluctuation bounds that can be used in
quantum hypothesis-testing converses.

\section*{Acknowledgements}

I acknowledge Ludovico Lami and Andreas Winter for discussions about this problem going back many years and Hami Mehrabi for more recent discussions. I acknowledge support from the Cornell University School of Electrical and Computer Engineering.

\section*{Statement on AI-assisted preparation}

OpenAI's ChatGPT Pro 6 Astra was used extensively in developing this manuscript,
including proposing proof strategies and mathematical arguments,
assisting with literature searches, reviewing the proofs, revising the
exposition, and preparing \LaTeX{} code and the manuscript. The author did not provide the
initial proof methods in advance. The author revised the manuscript and is responsible for verifying
the mathematical arguments, calculations, citations, and conclusions,
and for the final content of the manuscript.

\bibliographystyle{alpha}
\begingroup
\raggedright
\bibliography{pure_loss_strong_converse}
\endgroup

\appendix
\section{Proof of the logarithmic trace inequality}\label{app:trace}

We give a direct proof of Lemma~\ref{lem:log}. The exact $A$-weighted
logarithmic form appears explicitly in Sharma and Warsi
\cite[Lemma~5, Eq.~(A11)]{SharmaWarsi2012}, whose proof builds on the
related $B$-weighted trace inequalities discussed in
\cite[Lemma~1]{Audenaert2007} and
\cite[Theorem~11.18]{Petz2008}.
Only a scalar trace inequality is claimed; neither the weighted product
nor a logarithm compressed by the difference projector is assumed
positive as an operator.

\step{1. A strictly positive interpolation path.}
First suppose $A>0$ and $B>0$. Put
\begin{equation}
 \Delta\coloneqq A-B=\Delta_+-\Delta_-,\qquad
 P\coloneqq\{\Delta>0\}.
 \label{eq:appendix-difference}
\end{equation}
Then $P\Delta=\Delta P=\Delta_+$ and $P\Delta_-=0$.
For $0\leq s\leq1$, set $Z_s\coloneqq B+s\Delta=(1-s)B+sA$.
There is $m>0$ with $A,B\geq mI$, and thus $Z_s\geq mI$ uniformly in $s$.

\step{2. Differentiate the logarithm under an integral.}
For $z>0$, direct integration gives, as $b\to\infty$,
\begin{equation}
 \int_0^b\left(\frac1{1+x}-\frac1{z+x}\right)dx
 =\ln(1+b)-\ln(z+b)+\ln z\longrightarrow\ln z.
 \label{eq:scalar-log-integral}
\end{equation}
The finite-dimensional spectral theorem therefore implies
\begin{equation}
 \ln Z_s=\int_0^\infty
 \left[\frac1{1+x}I-(Z_s+xI)^{-1}\right]dx.
 \label{eq:matrix-log-integral}
\end{equation}
Differentiating $(Z_s+xI)(Z_s+xI)^{-1}=I$ gives
\begin{equation}
 \frac{d}{ds}(Z_s+xI)^{-1}
 =-(Z_s+xI)^{-1}\Delta(Z_s+xI)^{-1}.
 \label{eq:inverse-derivative}
\end{equation}
The norm of this derivative is at most
$\norm\Delta_\infty(m+x)^{-2}$, which is integrable in $x$ and uniform
in $s$. Differentiation under the integral, taking traces, and integration
in $s$ are therefore justified by this bound. Writing
$G_{s,x}\coloneqq(Z_s+xI)^{-1}$, we obtain
\begin{equation}
 \Tr[PA(\ln A-\ln B)]
 =\int_0^1\!\int_0^\infty
 \Tr[PA G_{s,x}\Delta G_{s,x}]\,dx\,ds.
 \label{eq:trace-integral}
\end{equation}

\step{3. Every integrand has a nonnegative decomposition.}
Fix $s\in[0,1]$ and $x>0$, and abbreviate $G\coloneqq G_{s,x}$.
Since $A=Z_s+(1-s)\Delta$, we have
\begin{equation}
 A=G^{-1}-xI+(1-s)\Delta,
 \qquad
 GAG=G-xG^2+(1-s)G\Delta G.
 \label{eq:resolvent-identities}
\end{equation}
Substitution into the integrand and $P\Delta=\Delta_+$ give
\begin{align}
\Tr[PAG\Delta G]
 &=\Tr[P\Delta G]-x\Tr[PG\Delta G]
      +(1-s)\Tr[P\Delta G\Delta G] \label{eq:resolvent-expansion-1}\\
&=\Tr[\Delta_+G]-x\Tr[PG\Delta G]
      +(1-s)\Tr[\Delta_+G\Delta G] \label{eq:resolvent-expansion-2}\\
&=\Tr[\Delta_+GAG]
      +x\Tr[\Delta_+G^2]-x\Tr[PG\Delta G]. \label{eq:resolvent-expansion-3}
\end{align}
The last equality uses the displayed formula for $GAG$.
Next, split $\Delta$ into its positive and negative parts and use cyclicity:
\begin{align}
\Tr[\Delta_+G^2]-\Tr[PG\Delta G]
 &=\Tr[G\Delta_+G]-\Tr[PG\Delta_+G]+\Tr[PG\Delta_-G] \label{eq:positive-negative-split-1}\\
&=\Tr[(I-P)G\Delta_+G]+\Tr[PG\Delta_-G]. \label{eq:positive-negative-split-2}
\end{align}
Combining these identities yields
\begin{equation}
 \Tr[PAG\Delta G]
 =\Tr[\Delta_+GAG]
  +x\Tr[(I-P)G\Delta_+G]
  +x\Tr[PG\Delta_-G]\geq0.
 \label{eq:positive-resolvent-decomposition}
\end{equation}
For example, $GAG\geq0$ and $\Delta_+\geq0$ imply
$\Tr\Delta_+GAG=\Tr\Delta_+^{1/2}GAG\Delta_+^{1/2}\geq0$.
For the other two terms, $G\Delta_+G\geq0$ and $G\Delta_-G\geq0$,
and the projections $I-P$ and $P$ are positive as well.
All three traces are consequently real and nonnegative. Integration
in~\eqref{eq:trace-integral} proves~\eqref{eq:log-trace} for faithful $A,B$.

\step{4. A singular first argument.}
For $A\geq0$, $B>0$, define
$A_\varepsilon\coloneqq A+\varepsilon I$ and
$B_\varepsilon\coloneqq B+\varepsilon I$ for $\varepsilon>0$.
Their difference is still $A-B$, so their strictly positive-difference
projector is the same $P$. The faithful case gives
\begin{equation}
 \Tr[P A_\varepsilon(\ln A_\varepsilon-\ln B_\varepsilon)]\geq0.
 \label{eq:faithful-regularization}
\end{equation}
The function $x\ln x$, extended by zero at $x=0$, is continuous on a
compact interval containing all eigenvalues under consideration.
Thus $A_\varepsilon\ln A_\varepsilon\to A\ln A$ in norm.
Since $B>0$, also $\ln B_\varepsilon\to\ln B$ and
$A_\varepsilon\ln B_\varepsilon\to A\ln B$ in norm.
Taking the trace and letting $\varepsilon\downarrow0$ proves
Lemma~\ref{lem:log}, including the reality assertion.\hfill$\square$

\section{Receiver--environment exchange as dark-mode parity}
\label{app:dark-parity}

This appendix verifies explicitly the identity stated in \eqref{eq:swap-dark-parity}:
\begin{equation}
 S_{BE}=(-1)^{N_{\rm dk}},
 \label{eq:swap-dark-parity-appendix}
\end{equation}
where $S_{BE}$ is the unitary that exchanges the receiver and environment
modes and
\begin{equation}
 N_{\rm dk}
 \coloneqq
 \sum_{i=1}^n d_i^\dagger d_i
 \label{eq:dark-number-appendix}
\end{equation}
is the total dark-mode photon-number operator. We use the spectral-calculus
definition $(-1)^{N_{\rm dk}}\coloneqq e^{i\pi N_{\rm dk}}$.

Recall the bright and dark creation operators
\begin{equation}
 a_{{\rm br},i}^\dagger
 \coloneqq
 \frac{b_i^\dagger+e_i^\dagger}{\sqrt{2}},
 \qquad
 d_i^\dagger
 \coloneqq
 \frac{b_i^\dagger-e_i^\dagger}{\sqrt{2}}.
 \label{eq:bright-dark-creators-appendix}
\end{equation}
By definition, receiver--environment exchange satisfies
\begin{equation}
 S_{BE}b_i^\dagger S_{BE}^\dagger=e_i^\dagger,
 \qquad
 S_{BE}e_i^\dagger S_{BE}^\dagger=b_i^\dagger.
 \label{eq:swap-creators-appendix}
\end{equation}
It follows immediately that
\begin{align}
 S_{BE}a_{{\rm br},i}^\dagger S_{BE}^\dagger
 &=
 \frac{e_i^\dagger+b_i^\dagger}{\sqrt{2}}
 =
 a_{{\rm br},i}^\dagger,
 \label{eq:swap-bright-appendix}\\
 S_{BE}d_i^\dagger S_{BE}^\dagger
 &=
 \frac{e_i^\dagger-b_i^\dagger}{\sqrt{2}}
 =
 -d_i^\dagger.
 \label{eq:swap-dark-appendix}
\end{align}
Thus the swap leaves every bright-mode creation operator invariant and
introduces a factor of $-1$ for every dark-mode creation operator.

To make the resulting parity relation explicit, consider a bright--dark
Fock basis vector
\begin{equation}
 \ket{\mathbf{m}}_{\rm br}\ket{\mathbf{r}}_{\rm dk}
 \coloneqq
 \left[
 \prod_{i=1}^n
 \frac{(a_{{\rm br},i}^\dagger)^{m_i}}{\sqrt{m_i!}}
 \right]
 \left[
 \prod_{i=1}^n
 \frac{(d_i^\dagger)^{r_i}}{\sqrt{r_i!}}
 \right]
 \ket{0}_{\rm br}\ket{0}_{\rm dk},
 \label{eq:bright-dark-fock-appendix}
\end{equation}
where
\begin{equation}
 \mathbf{m}\coloneqq(m_1,\ldots,m_n),
 \qquad
 \mathbf{r}\coloneqq(r_1,\ldots,r_n),
 \qquad \mathbf{m},\mathbf{r}\in\mathbb{N}_0^n.
 \label{eq:fock-indices-appendix}
\end{equation}
Here $\mathbb{N}_0\coloneqq\{0,1,2,\ldots\}$, and each vacuum ket in
\eqref{eq:bright-dark-fock-appendix} denotes the joint vacuum of $n$ modes.
The joint vacuum is invariant under receiver--environment exchange.
Using~\eqref{eq:swap-bright-appendix}--\eqref{eq:swap-dark-appendix},
\begin{align}
 S_{BE}
 \ket{\mathbf{m}}_{\rm br}\ket{\mathbf{r}}_{\rm dk}
 &=
 \left[
 \prod_{i=1}^n
 \frac{(a_{{\rm br},i}^\dagger)^{m_i}}{\sqrt{m_i!}}
 \right]
 \left[
 \prod_{i=1}^n
 \frac{(-d_i^\dagger)^{r_i}}{\sqrt{r_i!}}
 \right]
 \ket{0}_{\rm br}\ket{0}_{\rm dk}
 \label{eq:swap-fock-step-1}\\
 &\quad=
 (-1)^{\sum_{i=1}^n r_i}
 \ket{\mathbf{m}}_{\rm br}\ket{\mathbf{r}}_{\rm dk}.
 \label{eq:swap-fock-step-2}
\end{align}

On the other hand, by definition of the total dark-mode number operator,
\begin{equation}
 N_{\rm dk}
 \ket{\mathbf{m}}_{\rm br}\ket{\mathbf{r}}_{\rm dk}
 =
 \left(\sum_{i=1}^n r_i\right)
 \ket{\mathbf{m}}_{\rm br}\ket{\mathbf{r}}_{\rm dk}.
 \label{eq:dark-number-fock-appendix}
\end{equation}
Functional calculus therefore gives
\begin{equation}
 (-1)^{N_{\rm dk}}
 \ket{\mathbf{m}}_{\rm br}\ket{\mathbf{r}}_{\rm dk}
 =
 (-1)^{\sum_{i=1}^n r_i}
 \ket{\mathbf{m}}_{\rm br}\ket{\mathbf{r}}_{\rm dk}.
 \label{eq:dark-parity-fock-appendix}
\end{equation}
Comparing~\eqref{eq:swap-fock-step-2} and
\eqref{eq:dark-parity-fock-appendix}, the operators $S_{BE}$ and
$(-1)^{N_{\rm dk}}$ agree on every bright--dark Fock basis vector.
Both operators are bounded unitaries. Since these vectors form a complete
orthonormal basis of the joint $B^nE^n$ Fock space, equality on their
linear span extends by continuity to the whole space. It follows that
\begin{equation}
 S_{BE}=(-1)^{N_{\rm dk}}.
 \label{eq:swap-equals-dark-parity}
\end{equation}

Thus receiver--environment exchange is precisely total dark-mode
photon-number parity: vectors containing an even number of dark photons are
invariant under the swap, while vectors containing an odd number of dark
photons acquire a minus sign.

\end{document}